\documentclass[sigconf, nonacm]{acmart}

\usepackage{enumitem}
\usepackage{graphicx}
\usepackage{subcaption}
\usepackage{svg}
\usepackage{booktabs}
\usepackage{multirow}
\usepackage{tabularx}
\usepackage{xcolor}

\usepackage{mathtools}

\usepackage{amssymb}
\usepackage{amsthm}

\usepackage{algorithm}
\usepackage{algpseudocode}

\newtheorem{theorem}{Theorem}[section]

\newtheorem{definition}[theorem]{Definition}

\newtheorem{proposition}[theorem]{Proposition}

\usepackage{cleveref}

\crefname{theorem}{theorem}{theorems}
\Crefname{theorem}{Theorem}{Theorems}
\crefname{lemma}{lemma}{lemmas}
\Crefname{lemma}{Lemma}{Lemmas}
\crefname{definition}{definition}{definitions}
\Crefname{definition}{Definition}{Definitions}
\crefname{remark}{remark}{remarks}
\Crefname{remark}{Remark}{Remarks}
\crefname{proposition}{proposition}{propositions}
\Crefname{proposition}{Proposition}{Propositions}

\begin{document}

\title{Slipstream: Locality-Aware Graph Index Construction for Streaming Approximate Nearest Neighbor Search}

\author{Shubing Yang}
\affiliation{%
  \institution{University of Washington}
}
\email{sueyoung@uw.edu}

\author{Dongfang Zhao}
\affiliation{%
  \institution{University of Washington}
}
\email{dzhao@cs.washington.edu}

\renewcommand{\shortauthors}{Yang et al.}
\newcommand{\don}[1]{\textcolor{red}{[Dongfang: #1]}}

\begin{abstract}




Graph indexes are widely used for high-recall approximate nearest neighbor search (ANNS), but many real-time applications require streaming ANNS. In these real-time applications, continuously arriving embeddings must search the existing graph for candidate neighbors before updating graph edges, which makes repeated index construction a bottleneck for streaming ingestion workloads. 

We propose Slipstream, a new method that significantly reduces the computational cost of frequent insertions in graph indexes for ANNS. The core idea of Slipstream is exploiting the continuity in vector streams: the newly arrived point starts from promising candidates found during the previous insertion rather than searching from the entry point. More technically, Slipstream evaluates distinct subsets of starting candidates followed by an adaptive controller that narrows or widens the range according to the stream's stability. We further show that Slipstream is beyond heuristic: We derive an abstract model to characterize Slipstream's performance and analyze its theoretical bounds.

We have implemented Slipstream in two popular open-source libraries (Faiss, HNSWLib) and compared it with four baseline methods on five streaming vector datasets. Experimental results show that Slipstream achieves up to 30.8$\times$ higher end-to-end throughput than baselines while maintaining at least 0.95 recall@10.

\end{abstract}

\maketitle

\renewcommand\thefootnote{}\footnote{\noindent
The source code is available at \url{https://github.com/suey3141/Slipstream.git}.
}\addtocounter{footnote}{-1}

%

\section{Introduction}

Graph indexes are widely used for high-recall approximate nearest neighbor search (ANNS)~\cite{hnsw2018,NSG2017,diskann2019,anns_survey,qiu2025efficient,li2025scalable}, but many real-time applications now require streaming ANNS. In these applications, continuously arriving embeddings must search the existing graph for candidate neighbors before updating graph edges, which makes repeated index construction a bottleneck for streaming ingestion workloads. This streaming setting appears in workloads such as video analytics~\cite{tsn2016,i3d2017}, Retrieval-Augmented Generation (RAG)~\cite{rag2020,streamingqa,freshllms}, AI agent memory~\cite{memgpt}, and online recommendation~\cite{youtube2016,pinsage2018}. To keep new data searchable, newly arriving vectors must be incorporated soon after they are generated. The cost comes from the insertion procedure itself: each new point first searches the existing graph to find candidate neighbors and then updates graph edges~\cite{hnsw2018,parlayann}. When arrivals are frequent and query volume per insertion is low, this search can dominate online time, as shown in Figure~\ref{fig:breakdown}. This repeated construction cost therefore becomes the primary optimization target in streaming workloads.

\begin{figure}[t]
    \centering
    \begin{subfigure}[t]{0.48\columnwidth}
        \centering
        \includegraphics[width=\linewidth]{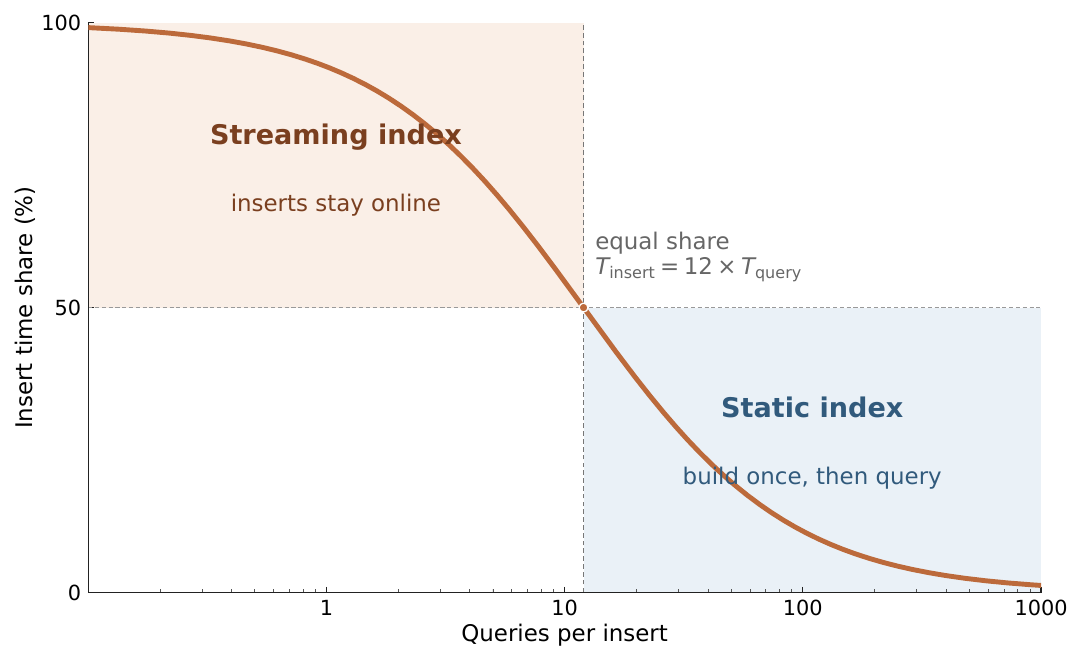}
        \caption{Insert time share.}
        \label{fig:breakdown}
    \end{subfigure}
    \hfill
    \begin{subfigure}[t]{0.48\columnwidth}
        \centering
        \includegraphics[width=\linewidth]{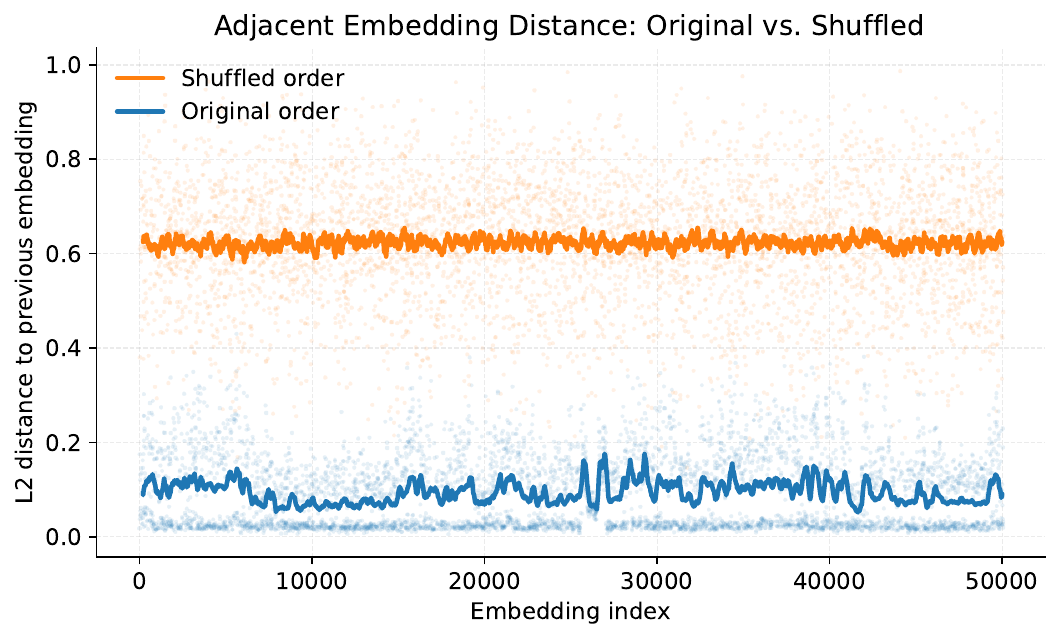}
        \caption{Temporal locality.}
        \label{fig:temporal_locality}
    \end{subfigure}
    \caption{(a) In insertion-heavy streaming workloads, index maintenance dominates time cost. (b) Consecutive points are substantially closer than after shuffling, indicating strong temporal locality.}
    \label{fig:intro_motivation}
\end{figure}

To address this bottleneck, we propose Slipstream, a new method that reduces the computational cost of frequent insertions in graph indexes for ANNS. The core idea of Slipstream is exploiting the continuity in vector streams: an arriving point starts from promising candidates found during the previous insertion rather than searching from the entry point. Standard graph insertion treats each new embedding as an independent search problem, causing repeated graph traversal and distance computation as embeddings arrive continuously. Many vector streams, especially video streams, contain adjacent semantically related items, and nearby frames or clips often share visual content~\cite{tsn2016,i3d2017,opticalflow2004}, so their embeddings tend to be close in vector space. Figure~\ref{fig:temporal_locality} visualizes this effect by comparing the distances between consecutive points in the original stream order with those in a randomly shuffled order. Since consecutive embeddings often target the same part of the graph, their insertion searches can overlap. Slipstream exploits this overlap so the next insertion can spend less computation finding candidate neighbors.

More technically, Slipstream conducts a proximity evaluation of distinct subsets of starting candidates and then uses an adaptive controller that narrows or widens the insertion range according to the stream's stability. For each stream segment, Slipstream caches the previous anchor, the candidates reached by insertion search, the selected neighbors, and their average distance scale. When a new point arrives, Slipstream evaluates whether different candidate subsets remain close enough to the new point to serve as reliable starting candidates. If the point remains within the reliable neighborhood, Slipstream seeds the next insertion with cached candidates and neighbors; otherwise, it discards the cache and follows standard insertion. The adaptive controller then adjusts the insertion range: it narrows the range on stable segments and widens it as stream stability weakens. The proximity check and adaptive controller determine both the seed set and the search budget of each insertion, while avoiding reuse when the previous insertion state no longer matches the current graph region.

We further show that Slipstream is beyond heuristic by deriving an abstract model to characterize its performance and analyzing its theoretical bounds. The abstraction model captures how the adaptive controller changes the insertion range under different stream stability conditions. A log scale approximation around this equilibrium yields the relation between average insertion width and the two controller settings. The recall bound combines a seed displacement argument, which keeps reused seeds near the new point's local neighborhood under small drift, with a fallback mixture argument that sends high drift cases to standard construction.

We implement Slipstream on HNSW using Faiss~\cite{douze2024faiss} and HNSWLib~\cite{hnsw2018}, and evaluate it on five video embedding streams including Kinetics~\cite{kay2017kineticshumanactionvideo}, BDD100K~\cite{bdd100k}, Epic-Kitchens~\cite{epickitchens}, Ego4D~\cite{ego4d}, and VIRAT~\cite{virat}. We compare Slipstream with four baselines: two vanilla HNSW implementations from Faiss and HNSWLib, Ada-ef~\cite{adaptivehnsw} based on HNSWLib, and DARTH~\cite{darth} based on Faiss. Slipstream achieves up to 30.8$\times$ higher end-to-end throughput than baselines while maintaining at least 0.95 recall@10.

This paper makes the following contributions:
\begin{itemize}

    \item We propose Slipstream, a new graph construction method for streaming ANNS that leverages embedding stream continuity to reduce the computational cost of frequent insertions. It carries promising candidates across locally coherent insertions, uses a proximity check to decide when they remain reliable, and adjusts insertion search width with an adaptive controller. (\Cref{sec:slipstream})

    \item We abstract Slipstream's performance into a controller model and theoretically analyze its recall bound. The model predicts segment averaged insertion width from the locality threshold and widening to narrowing step ratio, and the bound proves near target expected recall with a worst case standard construction floor. (\Cref{sec:theory})

    \item We implement Slipstream on HNSW and evaluate it on five video embedding streams, showing up to 30.8$\times$ higher end-to-end throughput than four state-of-the-art baselines while maintaining at least 0.95 recall@10.
    (\Cref{sec:evaluation})

\end{itemize}

\section{Related Work and Preliminaries}
\label{sec:related-work-prelim}

\subsection{Related Work}
\label{sec:related-work}

ANNS has been studied along four broad lines: hashing-based methods that map vectors into buckets~\cite{lsh1998, e2lsh2004, multiprobeLSH, falconn2015, c2lsh2012, qalsh2015, srs2014}; tree-based methods that recursively partition the space~\cite{kdtreebentley1975, flann2009tree, randomizedkdtree, scalableflanntree, covertree}; partition-based and quantization-based methods that compress vectors and prune via coarse cells~\cite{ivf2003, pq2010, opq2014, scann2020, lsq2018, rabitq2024}; and graph-based methods that perform greedy traversal on a proximity graph~\cite{nsw2014, hnsw2018, NSG2017, diskann2019}. Graph-based indexes currently offer the best recall and latency trade-off at scale and are widely deployed in production vector databases~\cite{milvus, manu2022, douze2024faiss, vbase2023, analyticdbv2020}; we focus on this category throughout the paper.

\paragraph{Graph-based approximate nearest
neighbor (ANN) indexes.}
\label{sec:related-graph}
The graph-based approximate nearest
neighbor (ANN) family began with Navigable Small World (NSW)~\cite{nsw2014} and was generalized by Hierarchical Navigable Small World (HNSW)~\cite{hnsw2018}, whose hierarchical small world structure underlies our work. Subsequent research has improved graph quality and connectivity~\cite{NSG2017, 2016efanna, mg2023, fanng2016, hcnng}, extended graph indexes to disk and tiered memory~\cite{diskann2019, lmdiskann, starling2024}, supported attribute-filtered search~\cite{filtereddiskann, acorn2024}, and scaled to industrial deployments~\cite{pinecone, weaviate, qdrant}. Closest to our setting, ParlayANN~\cite{parlayann} parallelizes graph construction by distributing batched insertions across threads, but each point's search is still treated as an independent operation.

\paragraph{Streaming and dynamic ANNS}
\label{sec:related-streaming}
To support continuous updates, prior work extends static ANN indexes along three axes. Hybrid update designs maintain a memory buffer and periodically merge into a base index~\cite{2021freshdiskann, ipdiskann, spfresh, cleann, topology}; organizations based on log structures and multiple tiers spread updates across levels to improve write throughput at the cost of query overhead~\cite{adaptivehnsw, 2025vectraflowcidr, lsmvec2024, milvus}; and studies of graph maintenance analyze how connectivity and recall degrade under sustained updates and propose repair routines~\cite{digra}. These approaches organize, schedule, or repair updates after insertion. Slipstream focuses on the insertion search itself and reuses search state across nearby arrivals when the stream exhibits locality.

\paragraph{Locality and adaptivity in indexing.}
\label{sec:related-locality}
Exploiting access locality is a foundational principle in systems and databases, from optimal replacement~\cite{belady1966} to result and intermediate result caching in large scale search~\cite{saraiva2001, cambazoglu2010caching}. In retrieval, embeddings produced by video and continuous sensing pipelines exhibit strong correlation between frames~\cite{tsn2016, i3d2017, slowfast2019, opticalflow2004}, providing a structural basis for reuse that has not previously been propagated into the ANN graph construction loop. Adaptive indexing has a long history in database systems, including adaptive merging~\cite{graefe2010adaptive}, and knob tuning based on learning~\cite{ottertune2017, cdbtune2019, qtune2019}. Within ANN specifically, Ada-ef~\cite{adaptivehnsw} adapts search width for each query and DARTH~\cite{darth} adjusts effort during insertion toward a recall target. Slipstream differs in two respects: its controller is driven by a locality signal rather than a global recall target, and reuse occurs at the granularity of search state rather than parameter values, composing naturally with the insertion procedure of the underlying graph index.

\begin{figure*}[t]
  \centering
  \includegraphics[width=0.85\linewidth, trim=0 0 0 40, clip]{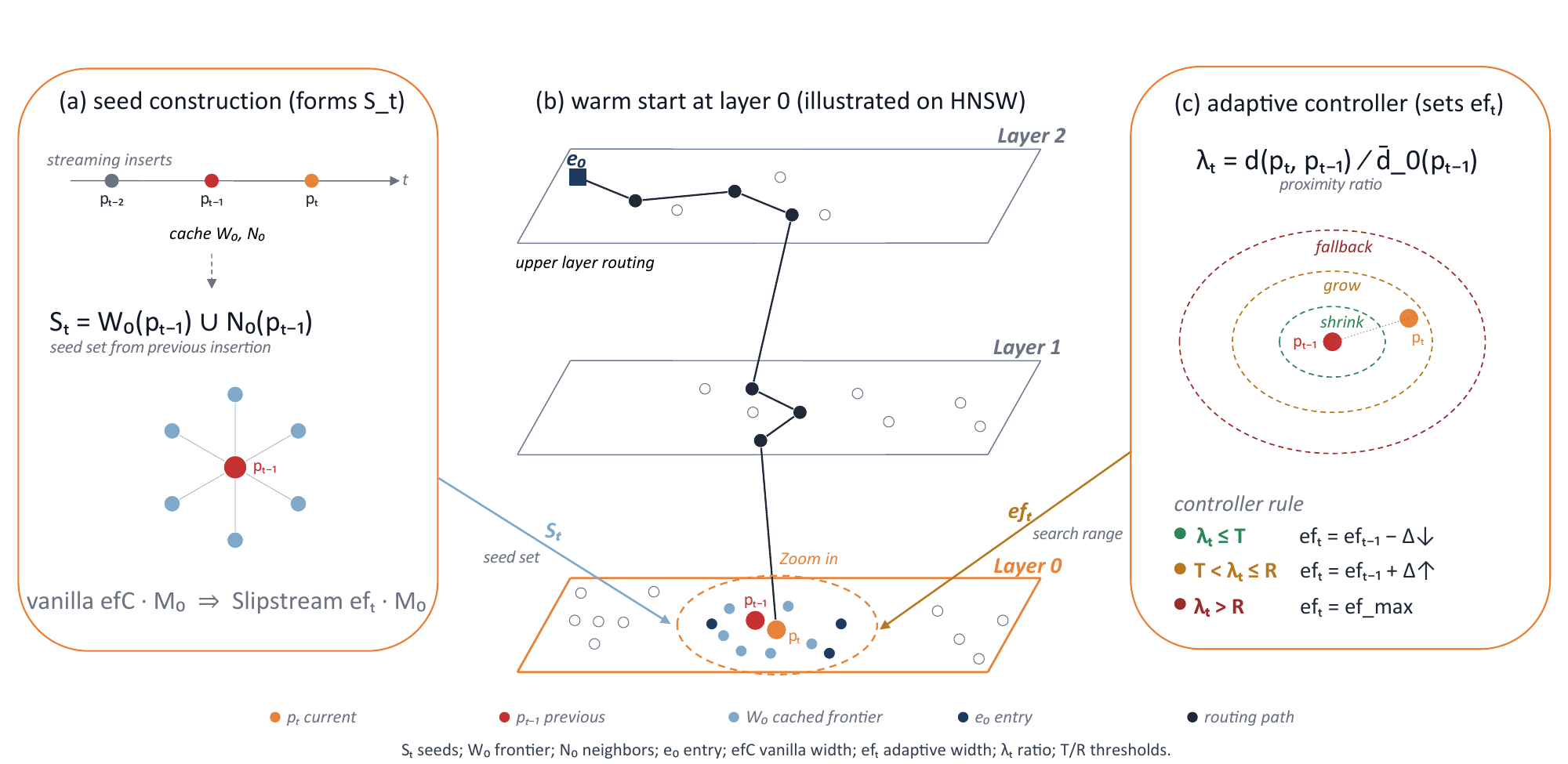}
  \caption{The overview of Slipstream. It illustrates Slipstream’s insertion pipeline: cached layer 0 candidates from the previous insertion initialize the current insertion, while an adaptive controller adjusts the search range according to normalized stream drift.}
  \label{fig:overview}
\end{figure*}

\subsection{Preliminaries}
\label{sec:prelim}

Slipstream targets graph ANN indexes that insert a new vector by
searching the current graph for candidate neighbors. We use
HNSW~\cite{hnsw2018} as the running instance because its insertion
procedure exposes the main cost component: routing in upper layers is
sparse, while layer~$0$ performs dense neighbor construction. This
section reviews the HNSW index (\Cref{sec:prelim_hnsw}) and its
insertion procedure (\Cref{sec:prelim_hnsw_insert}), highlighting the
layer asymmetry that motivates Slipstream.

\subsubsection{Hierarchical Navigable Small World (HNSW)}
\label{sec:prelim_hnsw}

HNSW~\cite{hnsw2018} indexes a dataset
$\mathcal{X}=\{x_1,\dots,x_N\}\subset\mathbb{R}^d$ as a stack of
proximity graphs $G_0, G_1, \dots, G_{L_{\max}}$. Layer $G_0$ contains every point and captures dense local neighborhoods, while higher layers are progressively sparser and provide shortcuts for long range routing. Each point $x$ is assigned a maximum level $\ell(x)=\lfloor -\ln(U)\cdot m_L\rfloor$ with $U\sim\mathrm{Uniform}(0,1)$, and appears in all layers from $G_0$ up to $G_{\ell(x)}$. With $m_L\propto 1/\ln M$, the level distribution decays exponentially, so only a small fraction of points reach the upper layers. Table~\ref{tab:hnsw_params} summarizes the main construction and query parameters used throughout the paper.

\begin{table}[t]
\centering
\caption{Main parameters of HNSW.}
\label{tab:hnsw_params}
\small
\begin{tabularx}{\linewidth}{lX}
\toprule
\textbf{Parameter} & \textbf{Meaning} \\
\midrule
$M$ & Maximum node degree in non-bottom layers ($\ell>0$). \\
$M_0$ & Maximum node degree at layer $0$, typically with $M_0>M$. \\
\texttt{efC} & Search width used during insertion-time graph search. \\
\texttt{efSearch} & Search width used during query-time graph search. \\
$m_L$ & Level multiplier controlling the level distribution. \\
$d(\cdot,\cdot)$ & Distance function used by the index. \\
$\ell(p)$ & Maximum sampled level of inserted point $p$. \\
\bottomrule
\end{tabularx}
\end{table}

Both insertion and query rely on a primitive, $\textsc{SearchLayer}(p,\allowbreak e,\allowbreak \texttt{ef},\allowbreak \ell)$, which greedily explores layer $G_\ell$ starting from the entry point $e$ and returns the $\texttt{ef}$ closest visited nodes to $p$, denoted $W_\ell(p)$. The width parameter $\texttt{ef}$ controls how many candidates are kept active during the walk, and is the main knob trading recall against cost.

\subsubsection{HNSW Insertion Procedure}
\label{sec:prelim_hnsw_insert}

Let $p$ be an arriving point with sampled level $\ell(p)$, and let $L$ be the current maximum level of the index. HNSW inserts $p$ in two stages, which differ in both cost and behavior. In the first stage, covering layers $L$ down to $\ell(p)+1$, HNSW starts from the global entry point $e_{\mathrm{entry}}$ and invokes $\textsc{SearchLayer}$ with $\texttt{ef}=1$ at each layer. This stage is a greedy routing procedure: it uses the returned node as the entry point for the next layer down and creates no edges.

In the second stage, covering layers $\hat{L}=\min(L,\ell(p))$ down to $0$, HNSW invokes $\textsc{SearchLayer}$ at each layer with the wider \texttt{efConstruction} budget to obtain $W_\ell(p)$. It then selects neighbors $N_\ell(p)=\textsc{SelectNeighbors}(p, W_\ell(p), M_\ell)$ subject to a degree bound $M_\ell$ for each layer ($M_0$ at layer~$0$, $M$ otherwise), connects $p$ bidirectionally to $N_\ell(p)$, and prunes neighbors whose degree exceeds the bound as needed. The closest selected neighbor becomes the entry point for the next layer down.

The two stages are asymmetric. Routing in upper layers operates on sparse graphs at $\texttt{ef}=1$ and creates no edges, while neighbor construction runs $\textsc{SearchLayer}$ with the wider \texttt{efC} budget on denser graphs. Since every inserted point reaches layer~$0$ and $M_0>M$, layer~$0$ construction dominates the insertion cost. Slipstream therefore preserves the standard upper layer procedure and modifies only the layer~$0$ insertion search.

\section{Slipstream}
\label{sec:slipstream}

Slipstream reduces streaming graph construction cost by carrying candidates across locally coherent insertions, using a proximity check to decide when reuse is reliable, and adapting the layer~0 insertion width. We describe it on HNSW, where $\textsc{BeamSearch}_0$ denotes the layer~0 insertion search and the upper layer entry follows standard HNSW routing; the same mechanism applies to graph ANN indexes whose insertion searches the current graph for candidate neighbors. \Cref{sec:slipstream-insertion} gives the insertion loop, \Cref{sec:lambda-model} defines the proximity gate, \Cref{sec:controller-rule} gives the controller, and \Cref{sec:complexity} analyzes cost.

\algrenewcommand\algorithmicrequire{\textbf{Input:}}
\algrenewcommand\algorithmicensure{\textbf{Output:}}

\subsection{Slipstream Insertion}
\label{sec:slipstream-insertion}

\begin{algorithm}[t]
\caption{\textsc{Slipstream}}
\label{alg:slipstream}
\begin{algorithmic}[1]
\Require stream batch $P=(p_1,\dots,p_B)$; index $\mathcal{H}$ with $M_0$, $M$, and parameters $\mathit{efC}$ and $\mathit{efR}$; fallback ratio $R$.
\Ensure updated $\mathcal{H}$.
\State $d_i \gets d(p_i,p_{i-1})$ for $i=2,\dots,B$
\State $\theta_{\mathrm{sc}} \gets 2R\cdot\operatorname{median}(d_2,\dots,d_B)$
\State split $P$ into segments $S_1,\dots,S_K$ at positions where $d_i>\theta_{\mathrm{sc}}$
\For{each $S_k$ \textbf{in parallel}}
\State $c \gets \mathbf{null}$
\Comment{per-segment cache $(a,\mathcal{C},\mathcal{N},\bar d_{\mathrm{nn}})$}
\ForAll{$p\in S_k$ in order}
\State $e \gets$ upper-layer greedy entry for $p$
\State $\lambda \gets d(p,c.a)/c.\bar d_{\mathrm{nn}}$ if $c\neq\mathbf{null}$ and $c.\bar d_{\mathrm{nn}}>0$, else $\infty$
\If{$\lambda \le R$}
\Comment{warm-start layer~0 insertion}
\State $\mathcal{S}\gets c.\mathcal{C}\cup c.\mathcal{N}$
\State $\mathcal{C}\gets\textsc{BeamSearch}_0(p,e,\mathcal{S},\ \mathbf{ef}=\mathit{efR})$
\Else
\Comment{fallback to standard layer~0 insertion}
\State $\mathcal{S}\gets\emptyset$
\State $\mathcal{C}\gets\textsc{BeamSearch}_0(p,e,\mathcal{S},\ \mathbf{ef}=\mathit{efC})$
\EndIf
\State $\mathcal{N}\gets\textsc{SelectNeighbors}(p,\mathcal{C},M_0)$
\State link $p\leftrightarrow\mathcal{N}$ and prune affected layer~0 adjacency lists as in standard HNSW
\State $c\gets(p,\ \mathcal{C},\ \mathcal{N},\ \overline{\mathrm{nn}}(p,\mathcal{N}))$
\EndFor
\EndFor
\State \Return $\mathcal{H}$
\end{algorithmic}
\end{algorithm}

\noindent
\Cref{alg:slipstream} shows the insertion loop of Slipstream. The method is based on the observation that consecutive stream points often arrive from nearby regions, so the previous layer~0 search can provide useful candidates for the next insertion. Slipstream first cuts each batch at large displacement jumps, as shown in lines~1--3, so that reuse is only attempted inside locally coherent segments. Each segment keeps a private cache containing the previous anchor point, its candidate set, selected neighbors, and local neighborhood scale. Here, $\overline{\mathrm{nn}}(p,\mathcal{N})=|\mathcal{N}|^{-1}\sum_{v\in\mathcal{N}}d(p,v)$ is the cached local scale used as $\bar d_0(p)$ in the proximity ratio below.

For a new insertion point $p$, Slipstream still uses the standard upper-layer greedy routing to obtain the layer~0 entry $e$. The change is confined to layer~0. Here, $\mathit{efR}$ denotes the reduced layer~0 beam width used on the warm-start path, while $\mathit{efC}$ is the standard construction width used by fallback insertion. If the locality test fails, Slipstream discards the cache and performs a standard HNSW layer~0 insertion with width $\mathit{efC}$. After either path, the cache is refreshed using the candidate set and neighbors of the newly inserted point. The next section defines the proximity ratio used to decide when cache reuse is safe.

\subsection{Proximity Ratio and Fallback}
\label{sec:lambda-model}

The reuse decision should not depend on an absolute distance threshold: the same displacement can be small in a sparse region but large in a dense region. Slipstream therefore normalizes inter-point displacement by the local neighborhood scale of the cached anchor.


\begin{definition}[Proximity ratio]\label[definition]{def:lambda}
For an insertion $p_t$, let $V_{t-1}$ be the vertex set already present before inserting $p_t$. If $p_t$ reuses the neighborhood cache of the preceding point $p_{t-1}\in V_{t-1}$, let $\mathcal{N}_0(p)$ denote the selected layer~0 neighbor set of $p$. Define the local neighborhood scale as
  $\bar d_0(p)
  =
  \frac{1}{|\mathcal{N}_0(p)|}
  \sum_{v\in\mathcal{N}_0(p)} d(p,v)$.
The \emph{proximity ratio} is
\begin{equation}\label{eq:lambda}
  \lambda_t
  =
  \frac{d(p_t,p_{t-1})}{\bar d_0(p_{t-1})}.
\end{equation}
\end{definition}

\noindent
The ratio $\lambda_t$ measures how far the new point moves in units of the cached point's local neighborhood radius. Small values indicate that the cached candidates are likely to overlap the new point's neighborhood; large values indicate that the stream has moved to a different region and cache reuse may miss relevant neighbors.

We set the fallback threshold by modeling the tail of $\lambda_t$. Near-duplicate frames are represented by a point mass at zero, while the remaining positive values follow a dynamic component $F_\Lambda$: $\pi_0\,\delta_0+(1-\pi_0)\,F_\Lambda$. Let $\Lambda$ denote the random variable corresponding to the observed ratios $\lambda_t$, with $\pi_0$ the mass of near-duplicate arrivals and $\delta_0$ a point mass at zero. For this dynamic component, we use an $\mathrm{Erlang}(2,\theta)$ distribution, which gives a compact tail model for two factors that affect reuse quality: inter-frame displacement and local density fluctuation. Its density and survival function are
\begin{align}
  f(\lambda) &= \tfrac{\lambda}{\theta^2}\,e^{-\lambda/\theta},
    \label{eq:erlang-pdf}\\
  S(\lambda) &\coloneqq P(\Lambda > \lambda)
    = \bigl(1+\tfrac{\lambda}{\theta}\bigr)\,e^{-\lambda/\theta}.
    \label{eq:survival}
\end{align}

The fallback ratio $R$ is chosen from this survival function. Given an upper bound $\theta_{\max}$ over the streams of interest and a target tail-probability budget $\varepsilon_1\in(0,1)$, we choose $R$ such that $S_{\theta_{\max}}(R)=\varepsilon_1$. For the full mixture, the tail probability is $(1-\pi_0)S_\theta(R)$, so this choice is conservative and gives $P(\Lambda>R)\le\varepsilon_1$ whenever $\theta\le\theta_{\max}$. Substituting $u=1+\lambda/\theta$ into \Cref{eq:survival} gives $u e^{-u}=\varepsilon_1/e$, which is inverted by the lower branch of the Lambert $W$ function.

\begin{definition}[Fallback ratio]\label[definition]{def:fallback}
The fallback ratio is
\begin{equation}\label{eq:R-def}
  R \;=\; \tau^{*}(\varepsilon_1,\theta_{\max}),\qquad
  \tau^{*}(\varepsilon,\theta)
  \;=\; -\theta\bigl[1 + W_{-1}(-\varepsilon\,e^{-1})\bigr].
\end{equation}
When $\lambda_t>R$, the warm-start cache is abandoned and standard HNSW insertion (with $\mathit{ef}=\mathit{efC}$) is used instead.
\end{definition}

\noindent
Because $S(\lambda)$ is increasing in $\theta$ for $\lambda>0$, any stream with $\theta\le\theta_{\max}$ satisfies $P(\Lambda>R)\le\varepsilon_1$. Thus, $R$ acts as a safety gate: when $\lambda_t>R$, Slipstream disables reuse and falls back to the standard insertion path. For insertions with $\lambda_t\le R$, reuse is allowed, but the layer~0 search width can still be adjusted according to local stability. This is the role of the adaptive controller.

\subsection{Adaptive Controller}
\label{sec:controller-rule}

The fallback rule makes a binary decision: reuse the cache or discard it. Among the insertions that pass this test, however, a fixed reduced width is not always appropriate. A stable segment can use a smaller layer~0 beam, while a segment with growing drift should use a larger beam before fallback becomes necessary. Slipstream therefore replaces the fixed warm-start width $\mathit{efR}$ with a per-segment adaptive width $\mathit{ef}_{\mathrm{cur}}$.

\Cref{alg:adaptive-controller} gives the controller logic used with the insertion loop in \Cref{alg:slipstream}. The controller does not change the structure of \Cref{alg:slipstream}; it only adds one state variable, $\mathit{ef}_{\mathrm{cur}}$, to each segment and substitutes this value for $\mathit{efR}$ on the warm-start path. After a successful warm start, the controller contracts the width when the proximity ratio is small and escalates it when the ratio moves closer to the fallback boundary. After a fallback, it resets the width to $\mathit{ef}_{\max}$ because the previous cache no longer represents the current local region.

\begin{algorithm}[t]
\caption{Adaptive controller for \textsc{Slipstream} (extends \Cref{alg:slipstream})}
\label{alg:adaptive-controller}
\begin{algorithmic}[1]
\Require controller $(\Delta_{\!\uparrow}, \Delta_{\!\downarrow}, T, \mathit{ef}_{\min}, \mathit{ef}_{\max})$.
\Statex \textbf{Per-segment state:}\ \ $\mathit{ef}_{\mathrm{cur}}\in[\mathit{ef}_{\min},\mathit{ef}_{\max}]$, initialized to $\mathit{efR}$ at segment start (cf.\ \Cref{alg:slipstream}, line~5).
\Statex \textbf{Site substitution:}\ \ in \Cref{alg:slipstream}, line~11, replace $\mathbf{ef}\!=\!\mathit{efR}$ with $\mathbf{ef}\!=\!\mathit{ef}_{\mathrm{cur}}$.
\Statex
\Procedure{OnWarmStart}{$\lambda$}
  \Comment{invoked after \Cref{alg:slipstream}, line~11}
  \If{$\lambda > T$}
    \State $\mathit{ef}_{\mathrm{cur}} \gets \min(\mathit{ef}_{\max},\ \mathit{ef}_{\mathrm{cur}} + \Delta_{\!\uparrow})$
      \Comment{escalate}
  \Else
    \State $\mathit{ef}_{\mathrm{cur}} \gets \max(\mathit{ef}_{\min},\ \mathit{ef}_{\mathrm{cur}} - \Delta_{\!\downarrow})$
      \Comment{contract}
  \EndIf
\EndProcedure
\Procedure{OnFallback}{}
  \Comment{invoked after \Cref{alg:slipstream}, line~14}
  \State $\mathit{ef}_{\mathrm{cur}} \gets \mathit{ef}_{\max}$
    \Comment{drift reset}
\EndProcedure
\end{algorithmic}
\end{algorithm}

In the batched-parallel implementation, each worker processes a contiguous segment and maintains its own controller state. This keeps adaptation local: stable segments quickly contract toward $\mathit{ef}_{\min}$, while segments with larger motion gradually move toward $\mathit{ef}_{\max}$. Let $\mathit{ef}_{\min}$ and $\mathit{ef}_{\max}$ be the lower and upper clamps, and let $\Delta_{\!\uparrow}$ and $\Delta_{\!\downarrow}$ be the escalation and contraction steps.

\begin{definition}[Escalation threshold]\label[definition]{def:escalation}
Let $\varepsilon_2\in(\varepsilon_1,1)$ be a second tail-probability budget. The escalation threshold is
\begin{equation}\label{eq:T-def}
  T \;=\; \tau^{*}(\varepsilon_2,\theta_{\max}),
\end{equation}
Since $\varepsilon_2>\varepsilon_1$, the survival threshold satisfies $T<R$, and $P(\Lambda>T)\le\varepsilon_2$ for any stream with $\theta\le\theta_{\max}$.
\end{definition}

The threshold $T$ separates stable reuse from cautious reuse. Values $\lambda_t\le T$ indicate that the cached neighborhood remains close to the new point, so the controller contracts the beam. Values $T<\lambda_t\le R$ still allow reuse, but indicate larger local drift, so the controller increases the beam before the fallback boundary is reached. The search width evolves as
\begin{equation}\label{eq:ef-update}
  \mathit{ef}_t =
  \begin{cases}
    \min\bigl(\mathit{ef}_{\max},\;
      \mathit{ef}_{t-1} + \Delta_{\!\uparrow}\bigr)
      & \text{if } T < \lambda_t \leq R,\\[2pt]
    \max\bigl(\mathit{ef}_{\min},\;
      \mathit{ef}_{t-1} - \Delta_{\!\downarrow}\bigr)
      & \text{if } \lambda_t \leq T,\\[2pt]
    \mathit{ef}_{\max}
      & \text{if } \lambda_t > R.
  \end{cases}
\end{equation}
Thus, $T$ controls how aggressively Slipstream reduces work within safe reuse cases, while $R$ remains the hard boundary that disables reuse.

\subsection{Complexity}
\label{sec:complexity}

We analyze construction cost by counting distance computations, which dominate graph ANN insertion. Slipstream changes only the layer~$0$ insertion search; the upper layer routing procedure remains identical to standard HNSW. We therefore separate the cost into three parts: batch segmentation, layer~$0$ construction, and upper-layer routing.

The segmentation step in \Cref{alg:slipstream} computes adjacent distances within a batch and splits the batch at displacement spikes. For a batch of $B$ vectors, this requires $O(B)$ distance computations and a median computation over the adjacent distances. Amortized over the batch, this adds constant work per insertion and does not change the asymptotic insertion cost.

For standard HNSW, a layer~$0$ insertion with construction width $\mathit{efC}$ costs
\begin{equation}
  C^{(0)}_{\mathrm{std}} = O(\mathit{efC}\cdot M_0)
\end{equation}
distance computations. The beam search keeps up to $\mathit{efC}$ active candidates and performs up to $O(\mathit{efC})$ expansions. Each expansion examines at most $M_0$ outgoing neighbors, giving the above bound. This is the dominant part of insertion because every point reaches layer~$0$, where the graph is densest.

A warm-start insertion replaces part of this expansion work with seed evaluation. Let $\mathcal{S}_t$ be the cached seed set used for insertion $p_t$. In \Cref{alg:slipstream}, $\mathcal{S}_t=c.\mathcal{C}\cup c.\mathcal{N}$, so $|\mathcal{S}_t|\le \max\{\mathit{efC},\mathit{ef}_{\max}\}+M_0$, which reduces to $\mathit{efC}+M_0$ under our setting $\mathit{ef}_{\max}\le\mathit{efC}$. Slipstream first evaluates the distance from $p_t$ to each seed and then runs layer~$0$ beam search with the current width $\mathit{ef}_t$. Thus the warm-start layer~$0$ cost is
\begin{equation}
  C^{(0)}_{\mathrm{ws}}(t)
  = O(|\mathcal{S}_t|+\mathit{ef}_t\cdot M_0).
\end{equation}
The seed term is a bounded additive overhead: each seed is evaluated once before the beam search begins. The expansion term is reduced from $O(\mathit{efC}\cdot M_0)$ to $O(\mathit{ef}_t\cdot M_0)$. Therefore, when $\mathit{ef}_t\ll \mathit{efC}$, the layer~$0$ search work decreases mainly through the smaller beam width.

The adaptive controller adds only constant work per insertion. It compares $\lambda_t$ with $T$ and $R$, updates $\mathit{ef}_{\mathrm{cur}}$ by a clamp operation, and resets the width after fallback. Hence the controller does not add a new asymptotic term. Its effect appears through the beam width $\mathit{ef}_t$ used on the warm-start path.

Slipstream alternates between warm-start and fallback insertions. Let $P_{\mathrm{ws}}$ be the warm-start hit rate in \Cref{eq:p-ws}, and let $E[\mathit{ef}]$ be the expected warm-start width induced by the controller in \Cref{eq:ef-closed}. On warm-start insertions, the expected layer~$0$ cost is $O(E[|\mathcal{S}_t|]+E[\mathit{ef}]\cdot M_0)$; on fallback insertions, Slipstream uses standard layer~$0$ insertion with cost $O(\mathit{efC}\cdot M_0)$. The expected layer~$0$ cost per insertion is therefore
\begin{equation}
  \bar{C}^{(0)}
  =
  O\!\left(
  P_{\mathrm{ws}}\bigl(E[|\mathcal{S}_t|]
  + E[\mathit{ef}]\cdot M_0\bigr)
  +
  (1-P_{\mathrm{ws}})\mathit{efC}\cdot M_0
  \right).
\end{equation}
Since $E[|\mathcal{S}_t|]\le \max\{\mathit{efC},\mathit{ef}_{\max}\}+M_0$, the seed evaluation cost remains bounded by the size of the previous search result and neighbor set. The expected saving is governed by two quantities: how often warm start is used, captured by $P_{\mathrm{ws}}$, and how small the controller keeps the warm-start width, captured by $E[\mathit{ef}]$.

The upper-layer cost is unchanged between standard HNSW and Slipstream. We keep it as a shared additive term $C_{\mathrm{upper}}(N,M)$, so the total expected insertion cost of Slipstream is $O(C_{\mathrm{upper}}(N,M)+\bar{C}^{(0)})$, while standard HNSW costs $O(C_{\mathrm{upper}}(N,M)+\mathit{efC}\cdot M_0)$. When layer~0 construction dominates the shared upper-layer term, the throughput improvement comes from replacing the standard layer~0 search cost with $\bar{C}^{(0)}$.

\section{Analysis and Bounds}
\label{sec:theory}

We analyze Slipstream along three axes: seed distance, controller equilibrium, and graph quality. \Cref{sec:theory-warmstart} bounds the distance from cached seeds to the new insertion; \Cref{sec:controller-equilibrium} characterizes the controller equilibrium and recall calibration relation; \Cref{sec:quality} relates insertion candidate quality to graph level query recall. Throughout this section, we assume that $d(\cdot,\cdot)$ is a metric and that $\bar d_0(p)>0$; if the local scale is zero, reuse is disabled for that insertion.

\subsection{Seed Quality}
\label{sec:theory-warmstart}

Standard HNSW initializes each layer~0 beam search from a single entry point obtained by greedy descent. Slipstream instead initializes the search with a seed set
$\mathcal{S}_t = W_0(p_{t-1})\cup\mathcal{N}_0(p_{t-1})$
from the previously inserted point, where $W_0(p_{t-1})$ is the beam search result set from $p_{t-1}$'s insertion and $\mathcal{N}_0(p_{t-1})$ is its selected neighbor set. In the analysis, $W_0(p_t)$ denotes the layer~0 candidate set returned by $\textsc{BeamSearch}_0$ for insertion $p_t$, corresponding to $\mathcal{C}$ in \Cref{alg:slipstream}.

The cached seeds remain close to the new insertion when the stream drift is small. For any cached vertex $v\in\mathcal{S}_t$, the triangle inequality gives
\[
d(p_t,v)\le d(p_t,p_{t-1})+d(p_{t-1},v).
\]
Averaging over the cached selected neighbors $\mathcal{N}_0(p_{t-1})$ gives
\begin{equation}\label{eq:seed-avg}
\frac{1}{|\mathcal{N}_0(p_{t-1})|}
\sum_{v\in\mathcal{N}_0(p_{t-1})}
d(p_t,v)
\le
(1+\lambda_t)\bar d_0(p_{t-1}).
\end{equation}
When $\lambda_t$ is small, the selected neighbors cached from $p_{t-1}$ remain close to $p_t$ at the local neighborhood scale. This distance bound gives the geometric basis for using a smaller beam width in the reuse regime.

\subsection{Controller Calibration}
\label{sec:controller-equilibrium}

Let $\bar e$ denote the segment averaged search width, measured as the average insertion search width over a stream segment. The controller analysis uses three relations. Model~1, $\rho(\bar e)$, relates recall to the segment averaged search width. Model~2, $\tau(\bar e)$, relates build time to the same quantity. Model~3, $\bar e(\beta^*,T)$, characterizes the equilibrium operating point selected by the controller, where $\beta^*=\Delta_{\downarrow}/(\Delta_{\uparrow}+\Delta_{\downarrow})$ is the drift balance and $T$ is the escalation threshold. Models~1 and~2 are fitted empirically in \Cref{sec:adaptive-calibration-section}. Here, we derive Model~3 and combine it with the recall model to obtain the resulting calibration relation.

The core derivation is an interior balance law. We reserve
$\lambda_t$ for the random proximity ratio defined in
\Cref{def:lambda} and use $\eta$ to denote a generic threshold
parameter; the controller instantiates $\eta = T$ via
\Cref{def:escalation}. Since fallback is already controlled by
\Cref{eq:R-def}, we derive the equilibrium by treating the
$O(\varepsilon_1)$ fallback mass $\{\Lambda > R\}$ as negligible
and focusing on the warm start regime. Let $e_{\min}$ denote the
floor of the segment averaged search width induced by the controller
clamp.

\begin{definition}[Quality pressure rate]
\label[definition]{def:quality-pressure}
For a threshold $\eta$ and a long run average search width
$\bar e$, define the \emph{quality pressure rate}
\begin{equation}\label{eq:pressure}
  \beta(\bar e,\eta)
  \;=\; P\bigl(\Lambda_t > \eta \,\bigl|\,
        \text{stream induces equilibrium } \bar e\bigr)\,.
\end{equation}
The conditioning is across the family of streams under
consideration: streams that drive the controller to the same
equilibrium $\bar e$ are identified, and the marginal tail
probability of $\Lambda_t$ over that subfamily is reported.
\end{definition}

\noindent
For fixed $\eta$, we assume $\beta(\cdot,\eta)$ is positive,
$C^1$ smooth, strictly decreasing on $(e_{\min},\infty)$,
satisfies $\beta(e_{\min},\eta) > \beta^*$, and obeys
$\lim_{\bar e\to\infty} \beta(\bar e,\eta) = 0$. For the local
calibration around an interior equilibrium, we further assume
$\partial_{\bar e}\beta(\bar e^{*},\eta)<0$ and
$\partial_{\eta}\beta(\bar e^{*},\eta)<0$.

\begin{proposition}[Equilibrium balance]
\label[proposition]{prop:equilibrium-balance}
Let the controller use threshold $\eta = T$, so that the per-step
transitions are: increment by $\Delta_{\!\uparrow}$ if
$T<\lambda_t\leq R$, decrement by $\Delta_{\!\downarrow}$ if
$\lambda_t\leq T$, and reset to $\mathit{ef}_{\max}$ if
$\lambda_t>R$. Under the rare fallback approximation
$P(\Lambda > R \mid \bar e) \approx 0$, the third regime contributes
negligibly and $P(T<\Lambda\leq R\mid\bar e)
\approx P(\Lambda>T\mid\bar e)$. The conditional mean drift at
$(\bar e,T)$ is
\[
  m(\bar e,T)
  \;=\; \Delta_{\!\uparrow}\,\beta(\bar e,T)
      - \Delta_{\!\downarrow}\,(1-\beta(\bar e,T)).
\]
Any interior equilibrium therefore satisfies
\begin{equation}\label{eq:quality-inv}
  \beta(\bar e^{\,*},T)
  \;=\; \frac{\Delta_{\!\downarrow}}%
             {\Delta_{\!\uparrow} + \Delta_{\!\downarrow}}
  \;=\; \beta^*\,.
\end{equation}
Under the smoothness and monotonicity assumptions above, this equilibrium
is unique and depends locally smoothly on $(\beta^*,T)$.
\end{proposition}

\begin{proof}
Conditioned on the segment averaged width $\bar e$, the controller has three possible updates. Before applying the rare fallback approximation, the expected per-step change is
\[
\begin{aligned}
  m_R(\bar e,T)
  &=
  \Delta_{\!\uparrow}P(T<\Lambda\le R\mid \bar e)
  -\Delta_{\!\downarrow}P(\Lambda\le T\mid \bar e) \\
  &\quad
  +(\mathit{ef}_{\max}-\bar e)P(\Lambda>R\mid \bar e).
\end{aligned}
\]
The rare fallback approximation drops the last term and replaces
$P(T<\Lambda\le R\mid \bar e)$ by
$P(\Lambda>T\mid \bar e)=\beta(\bar e,T)$. Since
$P(\Lambda\le T\mid \bar e)=1-\beta(\bar e,T)$, this gives
$m(\bar e,T)=\Delta_{\!\uparrow}\beta(\bar e,T)-
  \Delta_{\!\downarrow}(1-\beta(\bar e,T))$. An interior equilibrium has zero drift, so
\[
  0
  =
  \Delta_{\!\uparrow}\beta(\bar e^{*},T)
  -
  \Delta_{\!\downarrow}(1-\beta(\bar e^{*},T)).
\]
Rearranging gives
$(\Delta_{\!\uparrow}+\Delta_{\!\downarrow})\beta(\bar e^{*},T)
=\Delta_{\!\downarrow}$, which proves \Cref{eq:quality-inv}.
Strict monotonicity of $\beta(\cdot,T)$ gives uniqueness. Since
$\partial_{\bar e}\beta(\bar e^{*},T)<0$ at the interior equilibrium,
the implicit function theorem gives local smooth dependence on
$(\beta^*,T)$.
\end{proof}

The balance in \Cref{eq:quality-inv} defines the equilibrium
$\bar e^{\,*}(\beta^*,T)$ implicitly. Rather than postulate a global
functional form for $\beta(\bar e,\eta)$, we obtain a tractable
closed form by linearizing the equilibrium surface in log coordinates
around an operating point.

Write the excess effort above the controller floor as
$\tilde e := \bar e - e_{\min}$, fix a reference operating point
$(\beta^*_0, \eta_0)$ with corresponding equilibrium
$\tilde e^{\,*}_0 := \bar e^{\,*}(\beta^*_0,\eta_0)-e_{\min}>0$,
and define the local equilibrium elasticities
\begin{equation}\label{eq:elasticities}
  \alpha_1 :=
    -\,\frac{\partial \ln \tilde e^{\,*}}{\partial \ln \beta^*}
    \bigg|_{(\beta^*_0,\eta_0)},
  \qquad
  \alpha_2 :=
    -\,\frac{\partial \ln \tilde e^{\,*}}{\partial \ln \eta}
    \bigg|_{(\beta^*_0,\eta_0)}.
\end{equation}
Both $\alpha_1$ and $\alpha_2$ are well defined and positive by the
smoothness and monotonicity assumptions above. Concretely, $\alpha_1$
is the negative log slope of the equilibrium effort with respect to
the drift balance, and $\alpha_2$ is the negative log slope with
respect to the threshold.

A first order Taylor expansion of $\ln\tilde e^{\,*}$ in
$(\ln \beta^*,\ln\eta)$ gives
\begin{equation}\label{eq:loglin-equilibrium}
  \ln \tilde e^{\,*}(\beta^*,\eta)
  \;\approx\;
  \ln \tilde e^{\,*}_0
  \;-\; \alpha_1\,\ln(\beta^*/\beta^*_0)
  \;-\; \alpha_2\,\ln(\eta/\eta_0),
\end{equation}
valid in a neighborhood of $(\beta^*_0,\eta_0)$. Equivalently,
substituting $\eta = T$ and exponentiating,
\begin{equation}\label{eq:model3-derived}
  \bar e^{\,*}(\beta^*,T)
  \;\approx\;
  e_{\min} + A\,(\beta^*)^{-\alpha_1}\,T^{-\alpha_2},
\end{equation}
with
\begin{equation}\label{eq:model3-mapping}
  A \;=\; \tilde e^{\,*}_0 \,(\beta^*_0)^{\alpha_1}\,
          \eta_0^{\alpha_2}.
\end{equation}
The exponent ratio
\begin{equation}\label{eq:exponent-ratio}
  \alpha_2/\alpha_1
  \;=\;
  \frac{\partial \ln \tilde e^{\,*}/\partial \ln \eta}%
       {\partial \ln \tilde e^{\,*}/\partial \ln \beta^*}
  \bigg|_{(\beta^*_0,\eta_0)}
\end{equation}
is the marginal rate of substitution between threshold and drift balance along the equilibrium surface.

\paragraph{Iso-recall calibration.}
\label{par:iso-recall}
Let $\rho_t$ denote a target recall and let $\hat e_{M_1}(\cdot)$ be any continuous, strictly increasing recall to effort inverse derived from Model~1, so that $\rho\bigl(\hat e_{M_1}(\rho_t)\bigr)=\rho_t$. Setting $\bar e^{\,*}(\beta^*,T) = \hat e_{M_1}(\rho_t)$ in \Cref{eq:model3-derived}, the controller settings that target recall $\rho_t$ satisfy, to first order in log coordinates around the reference point,
\begin{equation}\label{eq:iso-recall-law}
  \alpha_1 \ln \beta^{*} \;+\; \alpha_2 \ln T
  \;=\; \ln A
        \;-\; \ln\!\bigl(\hat e_{M_1}(\rho_t) - e_{\min}\bigr).
\end{equation}
This relation is a straight line in $(\ln \beta^*,\,\ln T)$ coordinates with slope $-\alpha_1/\alpha_2$. The two knobs are realized through $\Delta_{\!\uparrow},\Delta_{\!\downarrow}$, which set $\beta^*$, and $\varepsilon_2$, which sets $T$.

Two auxiliary expressions are used later in the complexity and
calibration sections. First, the warm start hit rate is
\begin{equation}\label{eq:p-ws}
  P_{\mathrm{ws}}(R,\,L)
  \;=\; \frac{L-1}{L}\;
  \bigl[\pi_0 + (1-\pi_0)(1-S(R))\bigr]\,,
\end{equation}
where $(L-1)/L$ is cache availability within a segment and the
bracketed term is the probability that an available cache is usable.

Second, assume rare escalations ($\beta\ll 1$); the segment averaged
search width follows a clipped deterministic decay. For unit decrement
($\Delta_{\!\downarrow}=1$), let
$H=\mathit{efR}-\mathit{ef}_{\min}$ and $h=\min(L,H)$. Then
\begin{equation}\label{eq:ef-closed}
  E[\mathit{ef}]
  \;=\; \frac{1}{L}\left[
    \frac{h(2\,\mathit{efR} - h + 1)}{2}
    + (L - h)\,\mathit{ef}_{\min}
  \right].
\end{equation}
For general $\Delta_{\!\downarrow}$, the deterministic baseline is
\[
  E[\mathit{ef}]
  =
  \frac{1}{L}
  \sum_{i=0}^{L-1}
  \max\{\mathit{ef}_{\min},\,\mathit{efR}-i\Delta_{\!\downarrow}\}.
\]
This is the deterministic cost baseline used later for throughput accounting; it should not be confused with the local response surface in \Cref{eq:model3-derived}.

\subsection{Graph Quality Bound}
\label{sec:quality}

The warm start scheme modifies only the search phase of insertion,
not the selection phase. After the beam search returns a candidate
set $W_0(p_t)$, Slipstream applies the same
\textsc{SelectNeighbors} heuristic and bidirectional linking procedure
as standard HNSW. We relate graph quality to insertion quality in
three steps: candidate recall for each insertion, the effect of reuse
under locality, and graph level query recall.

\subsubsection{Insertion Time Candidate Quality}
\label{sec:quality-insertion}

The insertion quality signal used throughout this section is
\emph{insertion recall}. For an insertion $p_t$, let
$\mathcal{N}_0^{*}(p_t)$ be the true $M_0$ nearest neighbor set
of $p_t$ in $V_{t-1}$ and $W_0(p_t)$ the candidate set returned
by the beam search. The insertion recall is
\[
r(p_t)=
\frac{|W_0(p_t)\cap\mathcal{N}_0^{*}(p_t)|}
     {|\mathcal{N}_0^{*}(p_t)|}.
\]

Under Slipstream, warm start insertions ($\lambda_t\le R$) run beam
search at the current adaptive width
$\mathit{ef}_{\mathrm{cur}}\in[\mathit{ef}_{\min},\mathit{ef}_{\max}]$
from the initial frontier $\{e_0\}\cup\mathcal{S}_t$, where $e_0$
is the greedy descent entry used by the standard search and
$\mathcal{S}_t$ is the seed set inherited from $p_{t-1}$. We use the following monotone search abstraction: increasing the search width or adding initial seeds to the standard entry does not reduce insertion recall. Fallback insertions
reduce to standard HNSW at width $\mathit{efC}\ge
\mathit{ef}_{\max}$. Under this abstraction, for every insertion,
\begin{equation}\label{eq:quality-pointwise}
  r(p_t) \;\ge\;
  \begin{cases}
    r_{\mathrm{std}}(p_t;\mathit{ef}_{\mathrm{cur}}) & \text{if } \lambda_t \le R,\\
    r_{\mathrm{std}}(p_t;\mathit{efC}) & \text{if } \lambda_t > R.
  \end{cases}
\end{equation}
Taking the floor $\mathit{ef}_{\mathrm{cur}}=\mathit{ef}_{\min}$
gives the worst case warm start reference
$r_{\mathrm{std}}(p_t;\mathit{ef}_{\min})$.

\subsubsection{Reuse Under Locality}
\label{sec:quality-reuse}

The pointwise bound gives a conservative expected bound. The fallback
ratio in \Cref{def:fallback} caps the fallback probability at
$\varepsilon_1$ for streams with $\theta\le\theta_{\max}$, but this is
an upper bound on the high-width fallback branch and therefore cannot
be used as a lower bound on the fallback contribution. Taking
$\mathit{ef}_{\mathrm{cur}}=\mathit{ef}_{\min}$ as the worst case
adaptive width on the warm start branch gives
\begin{equation}\label{eq:quality-mixture}
  \mathbb{E}\bigl[r(p_t)\bigr]
  \;\ge\;
  \mathbb{E}\!\left[
    r_{\mathrm{std}}(p_t;\mathit{ef}_{\min})
  \right].
\end{equation}
Let $r_{\min,t}=r_{\mathrm{std}}(p_t;\mathit{ef}_{\min})$ and
$\Delta r_t=r_{\mathrm{std}}(p_t;\mathit{efC})-r_{\min,t}$. Then
\[
  \mathbb{E}[r(p_t)]
  \ge
  \mathbb{E}[r_{\min,t}]
  +
  \mathbb{E}\!\left[
    \mathbf{1}\{\lambda_t>R\}\Delta r_t
  \right].
\]

The bound above ignores the seed overlap already present before search. By \Cref{eq:seed-avg}, the cached selected neighbors have mean distance at most $(1+\lambda_t)\bar d_0(p_{t-1})$ from $p_t$. Thus, when $\lambda_t$ is small, the seed set is geometrically close to the target neighborhood, which explains why a small warm-start width can approach the candidate quality of a larger standard construction width. We use this observation as intuition; the formal lower bound relies only on the conservative candidate recall bound above.

\subsubsection{Impact on Query Recall}
\label{sec:quality-graph}

We now relate the insertion level bound to graph level query recall.
Let $G_N^{\mathrm{ws}}$ and $G_N^{\mathrm{std}}(\mathit{ef})$ denote
the layer~0 graphs constructed by Slipstream and standard HNSW
(insertion width $\mathit{ef}$), respectively, after $N$ streaming
insertions, using the same \textsc{SelectNeighbors} heuristic and
degree bound $M_0$. This step requires a graph level monotonicity
assumption: for two coupled builds with the same insertion order and
the same neighbor selection rule, if every insertion in one build has
candidate quality at least that of the corresponding insertion in the
other build, then its expected query recall is not lower.

Under this graph level monotonicity assumption, \Cref{eq:quality-pointwise}
implies the conditional lower bound
\begin{equation}\label{eq:recall-bound}
  \mathbb{E}_{q}\!\left[\mathrm{Recall}(q;\, G_N^{\mathrm{ws}})\right]
\;\geq\;
\mathbb{E}_{q}\!\left[\mathrm{Recall}(q;\, G_N^{\mathrm{std}}(\mathit{ef}_{\min}))\right].
\end{equation}
In expectation over the stream, the conservative insertion quality
reference is the floor width bound in \Cref{eq:quality-mixture}.

The bound has two readings. In the adversarial case, the controller
can remain clamped at $\mathit{ef}_{\min}$, giving the floor width
guarantee in \Cref{eq:recall-bound}. In the model guided case, for
streams consistent with the Erlang tail model at
$\theta\le\theta_{\max}$, the controller settles at the equilibrium
$\bar e^{*}(\beta^*,T)$ of \Cref{prop:equilibrium-balance}, and the
iso recall relation in the iso-recall calibration above ties this equilibrium
width to a target recall.

Suppose Slipstream operates on such a stream, and let
$\hat e_{M_1}(\rho_t)$ be the standard HNSW insertion width that
realizes target query recall $\rho_t$ via Model~1. If the controller
knobs $(\Delta_{\!\uparrow},\Delta_{\!\downarrow},\varepsilon_2)$
are placed on the iso recall line in \Cref{eq:iso-recall-law}, then
$\bar e^{*}(\beta^*,T)=\hat e_{M_1}(\rho_t)$. The fallback branch
contributes at width $\mathit{efC}\ge\bar e^{*}$, while the
discarded fallback mass contributes the same $O(\varepsilon_1)$ slack
as in the equilibrium derivation. This gives
\begin{equation}\label{eq:adaptive-recall}
  \mathbb{E}\bigl[\mathrm{Recall}(q;\,G_N^{\mathrm{ws}})\bigr]
  \;\gtrsim\; \rho_t \;-\; O(\varepsilon_1).
\end{equation}

Finally, fallback is applied exactly when the cache is least
trustworthy, namely when $\lambda_t>R$. These insertions are also the
cases where a small warm start width would be most likely to miss
neighbors. Routing them through standard HNSW at width $\mathit{efC}$
keeps the highest risk insertions on the baseline path. In practice,
this is why the observed recall of $G_N^{\mathrm{ws}}$ is often much
closer to $G_N^{\mathrm{std}}(\mathit{efC})$ than the floor width
bound alone suggests.

\section{Evaluation}
\label{sec:evaluation}

We evaluate Slipstream on five temporally ordered streaming workloads (\Cref{tab:datasets}). The evaluation validates the controller model and default operating point (\Cref{sec:adaptive-calibration-section}), compares throughput--recall trade-offs against baselines (\Cref{sec:throughput-recall}), decomposes streaming time (\Cref{sec:time-breakdown-section}), tests sensitivity to parameters (\Cref{sec:param-sensitivity}), and reports ablations and memory usage (\Cref{sec:ablation,sec:memory-usage}).

\textit{Setup.} All experiments were conducted on CloudLab c220g5 nodes equipped with two Intel Xeon Silver 4114 10-core CPUs (20 cores total at 2.20~GHz), 192~GB ECC DDR4-2666 memory, one 1~TB 7200~RPM 6G SAS HDD, and one Intel DC S3500 480~GB 6G SATA SSD. All code is written in C++ compiled with g++ 12, totaling over 3.5K lines.

\textit{Datasets.} We use five streaming video embedding workloads that preserve the original arrival order: Kinetics~\cite{kay2017kineticshumanactionvideo}, BDD100K~\cite{bdd100k}, Epic-Kitchens~\cite{epickitchens}, Ego4D~\cite{ego4d}, and VIRAT~\cite{virat}, as shown in Table~\ref{tab:datasets}. All embeddings are generated from video frames using CLIP~\cite{clip}. In every run, we initialize the index with a 100K vector base graph and ingest the remaining stream in batches of 5K embeddings. We evaluate every 10 insertion batches using 100 queries drawn sequentially from the dataset itself.

\textit{Baselines.} We compare Slipstream with four baselines: two vanilla HNSW implementations from HNSWLib~\cite{hnsw2018} and Faiss~\cite{douze2024faiss}, denoted as HNSWLib-Vanilla and Faiss-Vanilla, respectively; Ada-ef~\cite{adaptivehnsw}; and DARTH~\cite{darth}. Slipstream is evaluated with both HNSWLib and Faiss backends, denoted as Slipstream-HNSWLib and Slipstream-Faiss. Standard HNSW baselines use a fixed $\mathit{efSearch}=128$. DARTH uses Faiss with $\mathit{efSearch}=500$, target recall $0.95$, and a LightGBM predictor retrained every 10 batches using 1,000 training queries. Ada-ef uses HNSWLib with expected recall $0.95$ and recomputes its estimator between ef and recall every 10 batches.

\begin{table}[t]
\centering
\caption{Streaming datasets used in our evaluation.}
\label{tab:datasets}
\setlength{\tabcolsep}{3pt}
\begin{tabular}{@{}l l c c@{}}
\toprule
Dataset & Type & Dimension & \# Embeddings \\
\midrule
Kinetics & Action video & 512 & 2.00M \\
BDD100K & Driving video & 512 & 2.00M \\
EPIC-Kitchens & Egocentric video & 512 & 10.85M \\
Ego4D & Egocentric video & 512 & 3.88M \\
VIRAT & Surveillance video & 512 & 1.00M \\
\bottomrule
\end{tabular}
\end{table}

\subsection{Empirical Validation of the Controller Equilibrium Model}
\label{sec:adaptive-calibration-section}

This section empirically validates the controller equilibrium model by testing whether the predicted response surface matches measured behavior across driving, action, kitchen, egocentric, and surveillance streams. We sweep the controller drift balance $\beta^*$ and threshold $T$, measure the resulting segment averaged insertion width $\bar e$, and compare the observed surfaces with the form derived in \Cref{eq:model3-derived}. The goal is to test whether the model captures the qualitative and quantitative controller response across heterogeneous workloads.

On each workload we sweep a $5\times 3$ grid with $(\Delta_{\!\uparrow}, \Delta_{\!\downarrow})\in\{(5,1),\allowbreak(4,1),\allowbreak(2,1),\allowbreak(1,1),\allowbreak(1,2)\}$, i.e.\ $\beta^* \in \{\tfrac{1}{6},\tfrac{1}{5},\tfrac{1}{3},\tfrac{1}{2},\tfrac{2}{3}\}$, and $T\in\{0.5,\allowbreak 1.0,\allowbreak 1.5\}$. Each run consumes the leading $1$M~insertions of the stream with $R=2$ fixed and runs on a single thread, making the controller trajectory and insertion cost signal deterministic. For every grid point, we record the segment averaged width $\bar e$, insertion time, and $\mathrm{recall}@10$ at $\mathit{efSearch}\in\{32,64,128,256\}$.

The measured response surfaces are consistent with the derived model across workloads. Refitting the power law form in \Cref{eq:model3-derived},
\begin{equation}
  \bar{e}(\beta^*, T)
  \;=\; e_{\min} + A\,(\beta^*)^{-\alpha_1}\,T^{-\alpha_2},
  \label{eq:model3}
\end{equation}
independently on each workload yields the coefficients in \Cref{tab:model3-params}, with $R^2 > 0.97$ in linear space on all five streams (\Cref{fig:adaptive_goodness}). Although the prefactor $A$ varies by roughly two orders of magnitude, reflecting differences in dataset density and segment structure, the exponents remain large on every workload ($\alpha_1\!\ge\!1.9$ and $\alpha_2\!\ge\!3.4$). In particular, the large exponents indicate strong sensitivity to both drift balance and threshold, as predicted by the local linearization in log space in \Cref{eq:loglin-equilibrium}. This supports using \Cref{eq:model3-derived} as an approximation of the controller response that is independent of the workload. The fitted coefficients are descriptive: they validate the response surface, but are not used as controller parameters for individual workloads in the later evaluation.

\begin{figure*}[t]
  \centering
  \includegraphics[width=\textwidth]{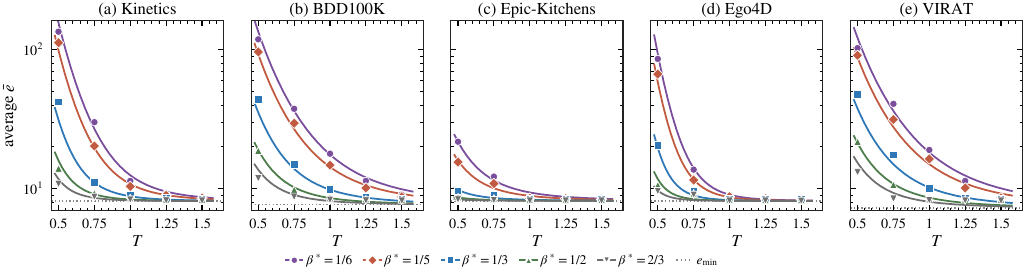}
  \caption{Empirical validation of the controller equilibrium model across five streams. Each panel compares measured and predicted segment averaged insertion width $\bar{e}$ over the $(\beta^*,T)$ controller grid.}
  \label{fig:adaptive_goodness}
\end{figure*}

\begin{table}[t]
  \centering
  \setlength{\tabcolsep}{4pt}
  \caption{Descriptive fits of the controller equilibrium model for each dataset
  (\Cref{eq:model3}) and the measured equilibrium width
  $\bar{e}_{\text{eval}}$ at the fixed global controller setting
  $(\Delta_{\!\uparrow}{=}4, \Delta_{\!\downarrow}{=}1,\,T{=}1.5)$ used in the remaining experiments.}
  \begin{tabular}{lrrrrrr}
    \toprule
    Dataset & $e_{\min}$ & $A$ & $\alpha_1$ & $\alpha_2$ & $R^2$ &
    $\bar{e}_{\text{eval}}$ \\
    \midrule
    Kinetics      & 8.16 & 0.036 & 2.68 & 5.08 & 0.969 & 8.59 \\
    BDD100K       & 7.70 & 0.193 & 2.20 & 3.68 & 0.981 & 8.77 \\
    Epic-Kitchens & 8.20 & 0.005 & 3.05 & 3.47 & 0.989 & 8.41 \\
    Ego4D         & 8.19 & 0.005 & 2.88 & 6.58 & 0.985 & 8.41 \\
    VIRAT         & 7.29 & 0.360 & 1.91 & 3.38 & 0.966 & 8.61 \\
    \bottomrule
  \end{tabular}
  \label{tab:model3-params}
\end{table}

After validating the response model, we use the same sweep to justify a single fixed global controller setting for the remaining experiments: $(\Delta_{\!\uparrow}, \Delta_{\!\downarrow},T,R)=(4,1,1.5,2)$. This setting is chosen from qualitative plateau behavior in the measured response surface, rather than by fitting a separate optimum for each workload. Along the $T{=}1.5$ row, insertion cost lies within $1\%$ of the grid minimum on all five workloads, whereas $T{=}0.5$ and $T{=}1.0$ require between $2$ and $3{\times}$ higher insertion cost on the harder workloads, BDD100K and Kinetics, for negligible recall gain. The large descriptive $\alpha_2$ values in \Cref{tab:model3-params} explain this behavior: the equilibrium width is highly elastic in the threshold, so a more permissive $T$ captures most of the available cost reduction. After fixing $T{=}1.5$, the five $\beta^*$ values produce insertion cost and recall within $1\%$ of one another on every dataset, so we use the conservative global setting $\beta^*=1/5$ ($\Delta_{\!\uparrow}{=}4$, $\Delta_{\!\downarrow}{=}1$). The measured equilibrium width $\bar e_{\mathrm{eval}}$ lies within a factor of $\sim\!1.4$ of the floor $e_{\min}$ for every dataset, and the realized recall stays within $0.01$ of the recall of the cheapest grid point at every standard $\mathit{efSearch}$, confirming that the chosen point lies on the recall plateau rather than in a degraded regime.

\subsection{Throughput--Recall Curve}
\label{sec:throughput-recall}

Figure~\ref{fig:throughput_recall} reports the throughput--recall@10 trade-off across five video embedding datasets, where each point corresponds to one operating configuration from the streaming sweep. Throughput is measured as inserted embeddings per second during the streaming phase and includes online insertion, offline maintenance, and query-time overheads. For a consistent sweep, all six methods, including Slipstream-HNSWLib, Slipstream-Faiss, and the four baselines, vary $\mathit{efC}$ with $M{=}16$ while keeping the query $\mathit{efSearch}$ fixed; Slipstream additionally uses $\mathit{efR}{=}32$ for adaptive reuse.

Across all datasets, Slipstream is the only method that combines tens of thousands of insertions per second with recall@10 above 0.95. This places the comparison in the intended regime: the gains come from reducing repeated insertion search rather than from relaxing the recall target. The configuration with the Faiss backend achieves the highest throughput on every dataset, reaching 42.5K, 45.4K, 40.8K, 39.0K, and 52.8K embeddings/s on Kinetics, BDD100K, Epic-Kitchens, Ego4D, and VIRAT, respectively. These points maintain recall@10 between 0.954 and 0.995. Compared with the fastest non-Slipstream baseline point on each dataset, Slipstream-Faiss improves throughput by 11.5--30.8x while remaining above the 0.95 recall threshold. This indicates that the main benefit is not simply choosing a lower quality operating point: Slipstream exposes high throughput configurations that remain within the high recall regime.

The Slipstream variant with the HNSWLib backend shows the same trend, although at lower throughput than the Faiss implementation. It reaches 29.9K embeddings/s on Kinetics, 15.1K on BDD100K, 11.2K on Epic-Kitchens, 12.3K on Ego4D, and 27.2K on VIRAT, all above 0.95 recall@10. Relative to the fastest baseline using HNSWLib, this corresponds to 3.2--23.1x higher throughput. The gap is largest on Kinetics and VIRAT, where reuse from the previous insertion preserves high recall while avoiding much of the insertion cost of repeated graph search.

Baselines exhibit the expected accuracy and throughput trade-off: increasing construction effort improves recall but quickly reduces streaming throughput. For example, vanilla HNSW and Faiss-HNSW can approach or reach perfect recall, but their fastest points are typically only 1--3.5K embeddings/s. Ada-ef and DARTH are designed to maintain high recall by adapting search effort as the index evolves. Under our streaming accounting, however, the adaptation work is charged during ingestion. As a result, their throughput remains close to or below the vanilla baselines on several datasets even though their recall is strong. The time breakdown in \Cref{sec:time-breakdown-section} attributes this behavior to periodic estimator recomputation in Ada-ef and predictor retraining in DARTH rather than to poor retrieval quality. In contrast, Slipstream reduces repeated search work directly on the insertion path and uses a lightweight controller, so its curve remains one to two orders of magnitude higher in throughput.

\begin{figure*}[t]
  \centering
  \includegraphics[width=\textwidth]{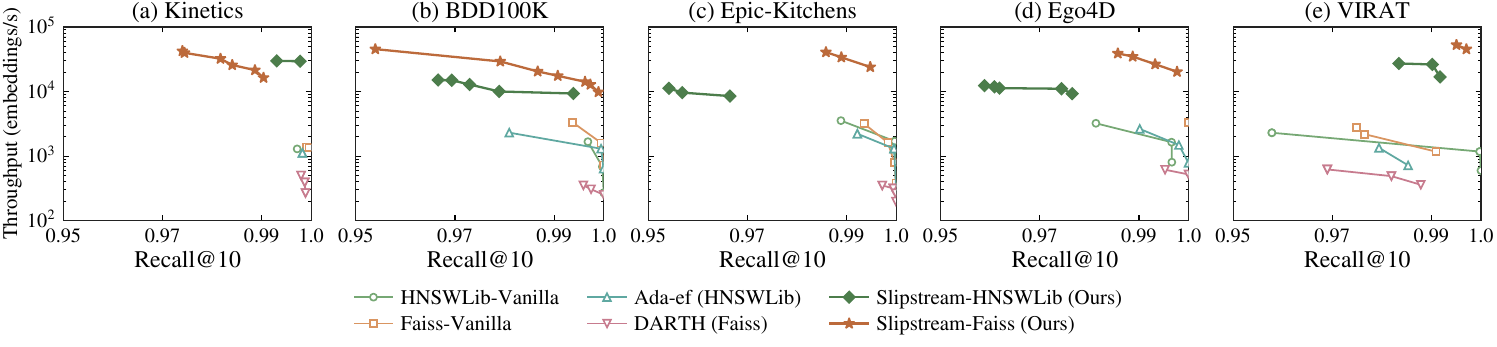}
  \caption{Streaming throughput versus recall@10 across the five video embedding workloads. Slipstream moves the frontier upward while remaining in the high recall regime.}
  \label{fig:throughput_recall}
\end{figure*}

\begin{figure*}[t]
  \centering
  \includegraphics[width=\textwidth]{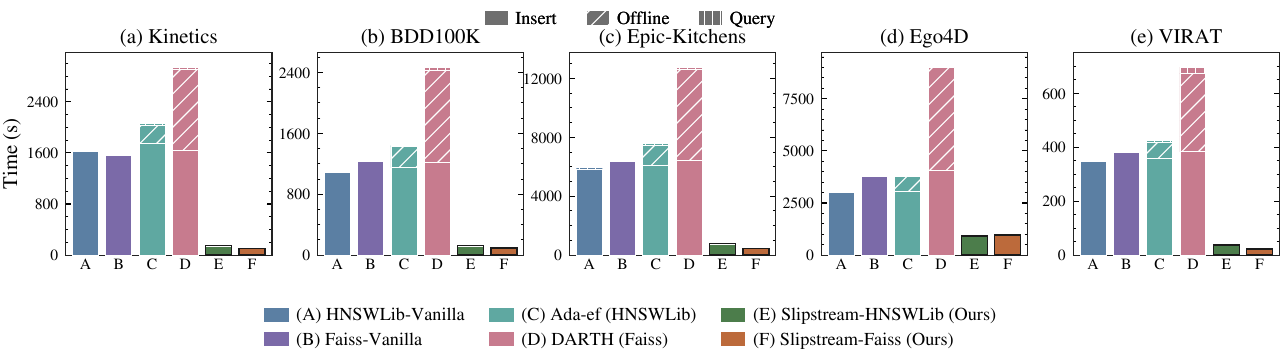}
  \caption{Operation time broken down into insertion, offline maintenance, and query processing. In this recall@10 streaming workload, DARTH and Ada-ef have insertion costs close to their vanilla backends, while predictor retraining or estimator recomputation adds maintenance overhead.}
  \label{fig:time-breakdown}
\end{figure*}

\subsection{Time Breakdown}
\label{sec:time-breakdown-section}

Figure~\ref{fig:time-breakdown} decomposes online elapsed time into insertion, offline maintenance, and query processing, excluding the cold start build. The goal is streaming cost accounting: it shows where each method spends time under our recall@10 workload. This setting is insertion dominated, so methods that mainly adapt query search effort, such as DARTH and Ada-ef, have limited room to offset insertion cost. Their mechanisms can be more beneficial when query cost is larger, for example at larger $K$.

Three patterns stand out. First, insertion dominates both vanilla baselines: HNSWLib-Vanilla and Faiss-Vanilla spend $99.6\%$--$99.8\%$ of online time in insertion. DARTH and Ada-ef leave the underlying HNSW insertion procedure largely unchanged, so their insertion components remain close to the corresponding vanilla backend. Second, DARTH spends a large fraction of runtime on offline maintenance ($42\%$ on VIRAT, $43\%$ on Kinetics, $48\%$--$49\%$ on BDD100K and EPIC-Kitchens, and $55\%$ on Ego4D), reflecting periodic retraining of its learned predictor. Third, Ada-ef has a smaller but visible maintenance component ($14\%$--$19\%$), due to periodic recomputation of the estimator between ef and recall. These costs are part of their recall control mechanisms rather than implementation artifacts, but under streaming ingestion they are paid in addition to vanilla-like insertion cost.

Slipstream reduces total online runtime by targeting the insertion bottleneck itself. Slipstream-Faiss takes $107.5$s on Kinetics, $92.8$s on BDD100K, $23.5$s on VIRAT, $441.0$s on EPIC-Kitchens, and $480.3$s on Ego4D, giving $14.6\times$, $13.3\times$, $16.3\times$, $14.4\times$, and $3.9\times$ reductions over Faiss-Vanilla. Slipstream-HNSWLib follows the same ordering but remains slower than Slipstream-Faiss. Unlike DARTH and Ada-ef, Slipstream reuses search state from the previous insertion and uses a lightweight controller to adjust the local insertion width, reducing the repeated search work that dominates HNSW construction. Query time is small across methods ($\le 4\%$ for vanilla and Slipstream-Faiss, up to $\sim\!11\%$ for Slipstream-HNSWLib and DARTH), so total streaming time is dominated by insertion.

\subsection{Parameter Sensitivity}
\label{sec:param-sensitivity}

\begin{figure*}[t]
    \centering
    \includegraphics[width=\textwidth]{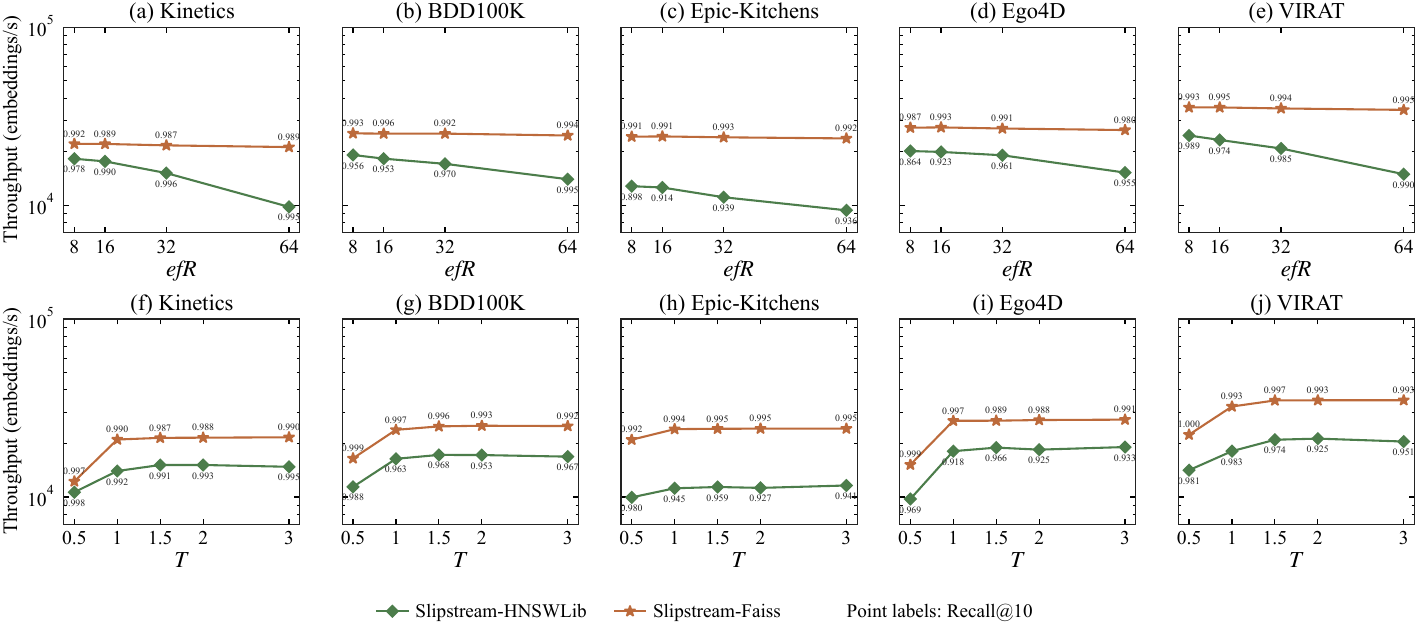}
    \caption{
    Sensitivity to Slipstream parameters. Columns correspond to the five workloads. The first row sweeps the initial warm start width $\mathit{efR}$, and the second row sweeps the controller threshold $T$. The y-axis reports streaming throughput, and labels next to points report recall@10.}
    \label{fig:param-sensitivity}
\end{figure*}

We next evaluate whether Slipstream depends on a narrow choice of its own parameters. We sweep the initial warm start width $\mathit{efR}\in\{8,16,32,64\}$ and the controller threshold $T\in\{0.5,1.0,1.5,2.0,3.0\}$. When sweeping $\mathit{efR}$, we fix $T=1.5$; when sweeping $T$, we fix $\mathit{efR}=32$. All other settings use the fixed global configuration justified in Section~\ref{sec:adaptive-calibration-section}.

Figure~\ref{fig:param-sensitivity} shows that Slipstream-Faiss is not sensitive to the initial value $\mathit{efR}$. Across the five workloads, changing $\mathit{efR}$ from $8$ to $64$ changes Faiss backend throughput by at most $4.3\%$ on any workload, while recall@10 stays at or above $0.980$. This is expected because $\mathit{efR}$ only initializes the active width in the adaptive version; after a few updates, the controller moves the active width toward the low cost region. Slipstream-HNSWLib shows a clearer recall and throughput trade-off: as $\mathit{efR}$ increases from $8$ to $64$, the geometric mean throughput drops from $18.6$K to $12.4$K embeddings/s, while mean recall@10 rises from $0.937$ to $0.974$.

The sweep over $T$ shows the main controller effect. A small threshold, $T=0.5$, makes the controller keep larger search widths and therefore lowers throughput. For Slipstream-Faiss, increasing $T$ from $0.5$ to $1.5$ raises geometric mean throughput from $17.0$K to $26.1$K embeddings/s, while mean recall@10 remains above $0.992$. For $T\ge 1.5$, the Faiss curves are nearly flat, which indicates that the controller has reached a stable low width region. Slipstream-HNSWLib follows the same broad pattern, with the best balance near the calibrated value $T=1.5$. Overall, the results show that the default setting is stable: moderate changes to $\mathit{efR}$ and $T$ do not remove Slipstream's throughput gains, and the flat regions indicate robustness rather than a hidden tuning point.

\begin{table}[t]
  \centering
  \footnotesize
  \setlength{\tabcolsep}{3pt}
  \caption{\textbf{Ablation A: component decomposition.} Each cell reports insert throughput ($10^3$ embeddings/s) / recall@10. Rows progress from vanilla HNSW to warm-start, fallback ratio $R{=}2$, and the full adaptive controller. Bold marks our full method.}
  \label{tab:ablation_a}
  \resizebox{\columnwidth}{!}{%
  \begin{tabular}{l c c c c c}
    \toprule
    Configuration & Kinetics & BDD100K & Epic-Kitchens & Ego4D & VIRAT \\
    \midrule
    Vanilla ($\mathit{efC}{=}200$) & 8.2/.995 & 14.6/.998 & 14.0/.999 & 26.7/.998 & 24.4/.994 \\
    Vanilla ($\mathit{efC}{=}32$) & 41.6/.947 & 56.4/.878 & 54.9/.904 & 102.3/.922 & 60.4/.926 \\
    \midrule
    Warm-start only & 7.9/.974 & 10.4/.958 & 9.9/.966 & 12.3/.971 & 10.0/.983 \\
    \quad $+$ Fallback ($R{=}2$) & 25.9/.984 & 27.8/.977 & 28.6/.989 & 32.5/.973 & 41.8/.991 \\
    \quad $+$ Adaptive (full) & \textbf{29.7/.947} & \textbf{32.8/.946} & \textbf{33.3/.956} & \textbf{49.8/.947} & \textbf{44.9/.970} \\
    \bottomrule
  \end{tabular}}
\end{table}

\begin{table}[t]
  \centering
  \footnotesize
  \setlength{\tabcolsep}{3pt}
  \caption{\textbf{Ablation B: fallback ratio $R$ sweep.} Each cell reports Slipstream throughput ($10^3$ embeddings/s) / recall@10.}
  \label{tab:ablation_b}
  \begin{tabular}{l ccc}
    \toprule
    Dataset & $R{=}1$ & $R{=}2$ & $R{=}\infty$ \\
    \midrule
    Kinetics & 22.8/.977 & 28.8/.951 & 14.0/.903 \\
    BDD100K & 28.8/.984 & 32.3/.949 & 16.4/.912 \\
    Epic-Kitchens & 29.4/.982 & 33.4/.957 & 19.5/.861 \\
    Ego4D & 44.3/.968 & 51.3/.948 & 26.7/.920 \\
    VIRAT & 40.9/.983 & 45.7/.965 & 18.9/.942 \\
    \bottomrule
  \end{tabular}
\end{table}

\begin{table*}[t]
\centering
\small
\setlength{\tabcolsep}{6pt}
\caption{Memory usage across 5 datasets.}
\label{tab:memory_combined}
\begin{tabular}{lrrrrr}
\toprule
\textbf{Method} & \textbf{Kinetics400} & \textbf{BDD100K} & \textbf{EPIC-Kitchens} & \textbf{Ego4D} & \textbf{VIRAT} \\
               & $N{=}2.00M$ & $N{=}2.00M$ & $N{=}10.85M$ & $N{=}3.88M$ & $N{=}1.00M$ \\
\midrule
\multicolumn{6}{c}{(a) Bytes per embedding} \\
\midrule
HNSWLib-Vanilla & 2292.5 & 2292.5 & 2288.5 & 2328.6 & 2295.7 \\
Faiss-Vanilla & 2204.5 & 2204.4 & 2193.2 & 2235.0 & 2219.4 \\
Ada-ef (HNSWLib) & 2295.9 & 2295.9 & 2266.5 & 2330.4 & 2303.3 \\
DARTH (Faiss) & 2236.9 & 2235.5 & 2231.7 & 2269.7 & 2240.5 \\

Slipstream-HNSWLib (Ours) & 2505.3 & 2505.3 & 2391.4 & 2399.1 & 2683.0 \\
Slipstream-Faiss (Ours) & 2200.0 & 2200.6 & 2197.1 & 2234.6 & 2209.1 \\

\midrule
\multicolumn{6}{c}{(b) Memory amplification (final RSS / streamed bytes; ideal=1.0)} \\
\midrule
HNSWLib-Vanilla & 1.119 & 1.119 & 1.117 & 1.117 & 1.121 \\
Faiss-Vanilla & 1.076 & 1.076 & 1.071 & 1.072 & 1.084 \\
Ada-ef (HNSWLib) & 1.121 & 1.121 & 1.107 & 1.118 & 1.125 \\
DARTH (Faiss) & 1.092 & 1.092 & 1.090 & 1.089 & 1.094 \\

Slipstream-HNSWLib (Ours) & 1.223 & 1.223 & 1.168 & 1.151 & 1.310 \\
Slipstream-Faiss (Ours) & 1.074 & 1.075 & 1.073 & 1.072 & 1.079 \\

\midrule
\multicolumn{6}{c}{(c) Index RSS (MB)} \\
\midrule
HNSWLib-Vanilla & 4373 & 4373 & 23689 & 8466 & 1908 \\
Faiss-Vanilla & 4205 & 4205 & 22702 & 8126 & 1845 \\
Ada-ef (HNSWLib) & 4379 & 4379 & 23462 & 8473 & 1915 \\
DARTH (Faiss) & 4265 & 4264 & 23102 & 8252 & 1862 \\

Slipstream-HNSWLib (Ours) & 4779 & 4779 & 24753 & 8722 & 2230 \\
Slipstream-Faiss (Ours) & 4196 & 4197 & 22742 & 8124 & 1836 \\

\bottomrule
\end{tabular}
\end{table*}

\subsection{Ablation Studies}
\label{sec:ablation}

We complement the end-to-end comparison with two ablations: one that isolates the contribution of each ingredient in Slipstream, and one that sweeps the fallback ratio $R$ to confirm that the choice $R{=}2$ (\Cref{def:fallback}) sits at the throughput knee.  Both use the five workloads of \Cref{tab:datasets} and report at $\mathit{efSearch}{=}32$ with $K{=}10$.  Where it applies, the controller uses $(\Delta_{\!\uparrow}, \Delta_{\!\downarrow},T,e_{\min})=(4,1,1.5,8)$ from \Cref{sec:adaptive-calibration-section}.

\Cref{tab:ablation_a} progressively adds each ingredient of Slipstream to a vanilla HNSW.  Two vanilla baselines bracket the design space: at $\mathit{efC}{=}200$ vanilla reaches the recall ceiling ($\ge .99$ on every workload) but is the slowest configuration ($8.2$--$26.7$K embeddings/s), while naively shrinking the construction beam to $\mathit{efC}{=}32$ buys a $2.5$--$5.1{\times}$ speedup at the cost of $5$--$12$ percentage points of recall, confirming that uniform budget reduction alone is not viable.  Warm-starting from the cached candidate set of the previous insertion recovers recall to $.958$--$.983$ but its throughput remains at the vanilla level
($7.9$--$12.3$K), because in the absence of a fallback every insert still pays a full $\mathit{efR}$ probe regardless of seed quality. Adding the fallback at $R{=}2$ is the single largest contribution: throughput jumps $2.6$--$4.2{\times}$ over warm-start alone, and recall \emph{improves} on every dataset (by up to $+2.3$ pp on Epic), since the fallback intercepts precisely the long-tail $\lambda$ events where the cached seed would otherwise misroute. The adaptive controller adds a further $7$--$53\%$ throughput on top of the warm-start-plus-fallback configuration at a cost of $2$--$4$ pp of recall, with the largest gains on streams with long stable runs that let $\mathit{efR}$ contract aggressively (Ego4D $+53\%$, Kinetics $+15\%$) and the smallest on workloads already operating near the controller floor $e_{\min}$ (VIRAT $+7\%$).

\Cref{tab:ablation_b} sweeps $R\in\{1,2,\infty\}$ for Slipstream. $R{=}2$ is the throughput maximum on all five workloads.  Below $R{=}2$, more than $20\%$ of inserts hit the fallback path and therefore pay the full $\mathit{efC}$ cost, with only marginal recall benefit because the warm-start seed was already locating the same neighborhoods.  Above $R{=}2$ the degradation is sharper: by $R{=}\infty$ (no fallback) throughput collapses by $1.7$--$2.4{\times}$, and recall drops most on streams with strong drift, with the Epic run falling to $.861$ ($-9.6$ pp below $R{=}2$) because scene changes are still warm-started from a stale anchor.

\subsection{Memory Usage}
\label{sec:memory-usage}

We report memory usage to verify that Slipstream's ingestion speedup does not rely on additional persistent index state. For each method, we measure process-resident memory after the streaming build completes, after raw input vectors and per-batch staging buffers are released, and with query and ground-truth evaluation structures excluded. \Cref{tab:memory_combined} reports bytes per vector, memory amplification, and absolute Resident Set Size (RSS) in megabytes.

On the Faiss backend, Slipstream introduces no measurable memory overhead beyond the matched vanilla baseline. Across the four $512$-dimensional workloads, Slipstream-Faiss stays within $0.5\%$ of Faiss-Vanilla in bytes per vector. This matches the design of the method: Slipstream changes how insertion search is initialized, but it does not change the node set, graph degree, or upper-layer hierarchy of the resulting HNSW index. Its cache, anchor, and controller counters are transient or negligible compared with the final graph.

On the HNSWLib backend, Slipstream shows a modest memory increase relative to HNSWLib-Vanilla, from $4.5\%$ on EPIC-Kitchens to $16.9\%$ on VIRAT in bytes per vector. Since the same algorithm shows no measurable overhead on Faiss, we attribute this difference to backend-specific bookkeeping in our HNSWLib port rather than to the algorithmic state required by Slipstream. Overall, the memory results support the intended conclusion: Slipstream improves ingestion throughput without adding a persistent auxiliary index or large memory structure.

\section{Conclusion}
\label{sec:conclusion}

This paper presented Slipstream, a method for reducing the cost of repeated index construction in streaming graph ANNS. Slipstream exploits continuity in vector streams by starting each new insertion from promising candidates discovered by previous insertions instead of restarting from the graph entry point. It evaluates distinct subsets of starting candidates with a proximity ratio and uses an adaptive controller to narrow or widen the insertion beam width according to stream stability. We further developed an abstract model and theoretical bounds to characterize when this reuse remains reliable. Implemented on HNSW using Faiss and HNSWLib, Slipstream achieves up to 30.8$\times$ higher end-to-end throughput than four baselines on five streaming vector datasets while maintaining at least 0.95 recall@10.



\bibliographystyle{ACM-Reference-Format}
\bibliography{reference}

\end{document}